\documentclass[runningheads,a4paper]{llncs}

\usepackage[T1]{fontenc}
\usepackage{amsmath,amssymb}
\usepackage{array}
\usepackage{booktabs}
\usepackage{graphicx}
\PassOptionsToPackage{hyphens}{url}
\usepackage{url}
\usepackage{cite}
\usepackage[hidelinks,hypertexnames=false]{hyperref}
\input{glyphtounicode}
\graphicspath{{figures/}{diagrams/}}

\newcommand{\R}{\mathbb{R}}
\newcommand{\E}{\mathbb{E}}

\newcommand{\rank}{\operatorname{rank}}
\newcommand{\one}{\mathbf{1}}

\newcolumntype{P}[1]{>{\raggedright\arraybackslash}p{#1}}
\spnewtheorem{assumption}{Assumption}{\bfseries}{\itshape}

\title{When Cross-Venue Agreement Is Not Price Discovery: Disclosure Frontiers for 24/7 Equity-Perpetual Oracles}
\titlerunning{Oracle Disclosure Frontiers}
\author{%
Donghwa Seo\inst{1} \and
Doohwi Cha\inst{2} \and
Seunghan Son\inst{3} \and
Juyeong Lee\inst{4} \and
Minjae Lee\inst{3} \and
Minsuk Sung\inst{5}}
\authorrunning{D. Seo et al.}
\institute{%
DS Investment \& Securities, Seoul, South Korea\\
\email{properitas95@kaist.ac.kr}
\and
Mirae Asset Securities, Seoul, South Korea\\
\email{doohwicha@gmail.com}
\and
Independent Researcher, South Korea\\
\email{qzpm5n@g.skku.edu}, \email{abraxasnz13@gmail.com}
\and
EY Consulting, Seoul, South Korea\\
\email{juyeong.lee@kr.ey.com}
\and
Department of Artificial Intelligence, Korea University, Seoul, South Korea\\
\email{minsuksung@korea.ac.kr}}

\begin{document}
\raggedbottom
\maketitle

\begin{abstract}
Crypto-listed equity perpetuals trade while the primary cash market is closed, yet
still need a mark for margin, funding, and liquidation. We model the closed-window
mark as the fixed point of an oracle operator with two blocks: external anchoring
and self/peer derivative reference. From marks and proxies alone the two are
observationally equivalent: every reduced form admits infinitely many topology
decompositions, and a path-law argument extends this to the full mark dynamics, so
lead-lag and information-share estimators have power equal to size. Disclosure
breaks the tie---disclosed diagonal adjustment identifies the normalized topology,
and disclosed support with forbidden anchors gives a row-level test that identifies,
falsifies, or leaves a positive-dimensional class under a rank condition and
finite-sample tolerance. Empirically, a disclosed OKX row survives pre-open
falsification while a pure-external baseline shows the test's limited power, and an
eight-week deep-closed panel with cash-reopen validation bounds the live-external
content of closure variance. Cross-venue agreement is not price discovery unless
disclosure or the cash reopen breaks the equivalence class.

\end{abstract}

\section{Introduction}

Perpetual futures and perpetual swaps are normally analyzed through funding,
replication, and no-arbitrage relations \cite{AngerisEtAl2023,AckererEtAl2025}.
For equity perpetuals listed in crypto markets, an additional market-design
problem appears: the derivative trades continuously while the primary
cash-equity venue is closed. At 03:00 New York time, on a weekend, or on a
US market holiday, an AAPL perpetual still needs a mark for margin, funding, and
liquidation. Such 24/7 equity perpetuals are now mainstream---Kraken markets the
first regulated tokenized-equity perpetuals, and OKX, Bybit, and Gate list stock
or stock-index perpetuals
\cite{KrakenXStocksPerps2026,OKXStockPerps2026,BybitTradFiPerps2026,GateIndexPrice2026}---and
their published rules build the closed-window mark from a stale cash close,
extended-hours prints, index futures, sector ETFs, tokenized-stock or RWA feeds,
overnight quote providers, the venue's own contract price, and peer derivative
marks, excluding stale components or carrying the last close forward while the
cash market is shut.

These published rules are concrete, not hypothetical. OKX mixes stock feeds,
RWA/tokenized-stock sources, peer equity perpetuals, its own contract, and other
exchanges' indices, reweighting toward crypto-native sources off-hours
\cite{OKXStockPerps2026}; Bybit excludes stale components and clamps the mark
\cite{BybitTradFiPerps2026}; Gate combines stock-token and U.S.\ stock prices with
overnight-quote or carry-forward fallbacks \cite{GateIndexPrice2026}; Binance and
Hyperliquid build marks from index, book, contract, and median components with
fallbacks \cite{BinanceMarkPrice2026,HyperliquidOracle2026}; and Kraken's regulated
tokenized-equity perpetuals confirm the product class \cite{KrakenXStocksPerps2026}.
In every case the mark mixes genuinely external feeds with the venue's own and peer
derivative prices---the structure the rest of the paper formalizes---as
Table~\ref{tab:disclosure} surveys.

\begin{table}[t]
\caption{Disclosed closed-window mark inputs for documented venues (June 2026).}
\label{tab:disclosure}
\centering
\footnotesize
\begin{tabular}{P{0.15\linewidth}P{0.57\linewidth}P{0.17\linewidth}}
\toprule
Venue & Disclosed closed-window inputs & Anchoring \\
\midrule
OKX & Hyperliquid oracle, OKX perpetual, Binance index, external vendor feeds (Pyth/Ondo/dxFeed) & external + self/peer \\
\midrule
Binance & price index, funding, order book, own contract, data fallbacks & external + self \\
\midrule
Bybit & index components with off-hours stale exclusion, mark clamp & external (clamped) \\
\midrule
Gate & Gate stock-token, U.S. stock price, own contract ($3/4$ underlyings), overnight quotes else carry-forward & external + self \\
\midrule
Hyperliquid & median of validator oracle, local book, external perpetual & dual external/peer \\
\bottomrule
\end{tabular}
\par\smallskip
{\footnotesize OKX component snapshots are frozen under
\texttt{artifact/source\_snapshots/} (June 2026); Bybit, Gate, and Kraken rules
were rechecked 26 June 2026.}
\end{table}

The frozen AAPL audit gives one documented component-row replay: the public OKX
AAPL-USDT index reports Hyperliquid Oracle $0.343$, OKX Linear Perpetual $0.057$,
Binance Index $0.171$, and external vendor weight $0.429$ (Pyth, Ondo, dxFeed),
with identical weights for NVDA, TSLA, and AMZN and a $0.260$ bps replay residual.
Because Hyperliquid provenance is not pinned by the public endpoint, two codings
are carried throughout: the as-parsed peer coding (derivative weight $0.571$) and
the conservative Hyperliquid-as-external coding ($0.228$). Cash-flow boundaries are
kept separate from topology via $\tilde z_t=T_tz_t+c_t$, with $T_t$ for split,
currency, and numeraire and $c_t$ for dividends and corporate actions; where these
are not fixed by the rule the residual is a boundary effect, not evidence about
$W$, so the empirics restrict to windows where they are flat or documented.

The stakes are not cosmetic: the closed-window mark sets margin, funding, and
liquidation, so for hours at a time it is the operative price for every leveraged
position. If the mark is a direct estimate of external equity value the contract
is a faithful 24/7 shadow of the cash market; but if the rule lets venues
reference their own or each other's derivative marks, the mark becomes the
equilibrium of a feedback loop---marks that agree because they read one another,
not because they track an outside price. Such an endogenous fixed point can sit
far from the eventual cash reopen yet look perfectly orderly---low dispersion,
high proxy correlation---until the cash market reopens at a level the network
never anchored to, so telling the two regimes apart is a first-order risk question
this paper formalizes.

The question is mathematical before it is empirical: once the rule permits self-
or peer-reference, the mark is the equilibrium of a networked oracle operator and
the object that matters is the mark's effective exposure to external anchors
through the whole reference network. The natural
instinct---test anchoring with the multi-market price-discovery toolkit, or read
cross-venue agreement as health---fails by construction, because agreement is
exactly what an endogenous network manufactures and the price-discovery estimands
take the same value from purely external to purely self-referential topologies. No
amount of mark or proxy data breaks the tie, since it is an observational
equivalence rather than a finite-sample artifact; what breaks it is information
from outside the price data---the mechanism the venue discloses and the tradable
level the cash market prints when it reopens.

We model the closed-window mark as an oracle fixed point and establish when it
exists, is unique, and is stable. We then prove an estimator-level
impossibility---infinitely many topologies share one reduced form, and a path-law
argument extends this to the full conditional mark dynamics, so information shares,
lead-lag, and Granger statistics cannot separate external anchoring from peer
consensus---together with a matching possibility result: disclosed diagonal
adjustment identifies the normalized topology, and disclosed support, signs, and
forbidden feeds turn the same algebra into a row-level test that identifies a row,
falsifies the rule, or leaves a positive-dimensional class under a minimal-rank
condition and a finite-sample tolerance bound. We read documented venue rules
through this lens and bring a frozen, reproducible battery---a four-underlying
pre-open row falsification, an eight-week deep-closed panel, and an archival
cash-reopen study---to bear, consistent with the claim that disclosure, not more
price data, identifies a closed-window mark while stopping short of recovering the
topology.

The paper is organized around that boundary. Section~2 places the question against
the price-discovery, oracle-manipulation, and peer-effect literatures by estimand.
Section~3 sets up the closed-window oracle model and its fixed point and shows that
funding sets only convergence speed, not the target. Section~4 proves the
non-identification and path-law impossibility, the matching diagonal-adjustment and
disclosure-frontier possibility results, and the cash-reopen decomposition that
validates anchoring. Section~5 illustrates the mechanism on a synthetic network,
and Section~6 brings the frozen empirical battery to bear; the paper concludes with
the mechanism-design implication.

\section{Related Literature}

Perpetual pricing studies funding, replication, and no-arbitrage against a
continuously tradable underlying \cite{AngerisEtAl2023,AckererEtAl2025}; equity
perpetuals break that premise because the cash market is closed for most of the
contract's life, and the product is now mainstream
\cite{KrakenXStocksPerps2026,OKXStockPerps2026,BybitTradFiPerps2026,GateIndexPrice2026}.
That literature explains how funding ties a perpetual to a given index but
is silent on how the index is formed with no cash price. The
blockchain-oracle literature studies how off-chain values enter a chain and how
that bridge is attacked
\cite{EskandariEtAl2021,QinEtAl2021};
closest to us, \cite{YangKlagesMundtGudgeon2023} recovers an off-chain level
to flag anomalies, asking whether a price is right where we ask which feeds
a mark is built from---an oracle can be honest yet leave its anchoring
unidentified.

The classical price-discovery toolkit---lead-lag, common-factor, and
information-share measures
\cite{Hasbrouck1995,GonzaloGranger1995}---presupposes a shared
contemporaneous efficient price that a closed cash market lacks, and our path-law
result shows these exact estimands are invariant across observationally equivalent
topologies. After-hours discovery is small and noisy
\cite{BarclayHendershott2003}, so a mark built from
extended-hours prints is doubly removed from a tradable price and the cash reopen,
not the closed-window data, is the external event our validation leans on. Our
impossibility is the closed-market analogue of the reflection problem
\cite{Manski1993,BramoulleEtAl2009} and the near-unit consensus mode is the oracle
counterpart of naive network learning \cite{GolubJackson2010}; the overnight
stale-price regime that time-zone arbitrage exploits \cite{DonnellyTower2008} is
exactly where our row diagnostic is most fragile.

Across these literatures the estimand never matches the question: price discovery
estimates an information share, the attack literature an exploitability margin, and
the reflection literature a non-separation result, none pinning down the object a
venue's documentation can. That object is the admissible disclosure set
$\mathcal A_v(R)$, the oracle rows consistent with both the observed reduced form
and the venue's written restrictions, which collapses to a point, falsifies the
rule, or stays positive-dimensional. Reframing the question around
$\mathcal A_v(R)$ turns disclosed support, forbidden anchors, weights, and
fallbacks---not more price data---into the lever that identifies a closed-window
mark.

\section{Closed-Market Oracle Model and Fixed Points}
This section turns those disclosures into a single operator. The object to model
is not a price but a rule: each venue publishes an index that mixes external
anchors with same-underlying derivative marks, and the traded mark is pulled
toward that index. Writing the mixture as a linear map makes the closed-window
mark the solution of a fixed-point equation whose two blocks---external anchoring
through $B$ and derivative self/peer reference through $W$---are exactly what an
observer must later tell apart.

Fix one underlying equity and one cash-market closure window. There are $V$
venues or mark sources. Let $m_t\in\R^V$ denote same-underlying perpetual marks
and $z_t\in\R^Q$ denote external anchors observable during the closure. Each
venue computes an index or mark target
\begin{equation}
  i_t = W_t m_t + B_t z_t + \eta_t.
  \label{eq:index}
\end{equation}
The matrix $W_t\in\R^{V\times V}$ is the derivative self/peer-reference
matrix. The matrix $B_t\in\R^{V\times Q}$ maps external anchors into venue
indices. The residual $\eta_t$ captures unobserved oracle noise, discretionary
adjustments, stale feeds, filtering, and measurement error. Mark dynamics pull
toward the index:
\begin{equation}
  dm_t=-K_t(m_t-i_t)\,dt+\Sigma_t\,d\epsilon_t.
  \label{eq:dynamics}
\end{equation}
Substitution gives
\begin{equation}
  dm_t =
  -K_t\left[(I-W_t)m_t-B_tz_t-\eta_t\right]dt
  +\Sigma_t\,d\epsilon_t .
  \label{eq:substituted}
\end{equation}

Each object has a plain reading: $B$ is the oracle quoting the outside world
(cash, futures, ETF, or vendor feeds produced independently of the venue's
derivatives), $W$ is the oracle quoting other oracles (self- and peer-reference),
$K$ is how hard funding and mark control pull the mark toward the index, and
$\eta$ collects discretion, clamps, filters, and stale-feed substitutions. The
identification problem is exactly the split between $B$ and $W$: the realized
mark is $m^*=(I-W)^{-1}(Bz+\eta)$ (Lemma~\ref{prop:existence}), so a
low-$W$, high-$B$ venue is externally anchored, while a high-$W$, low-$B$ venue
can display the same cross-venue agreement even though that agreement is
manufactured by the derivative network rather than discovered from an external
price.

\begin{assumption}[Active affine regime]\label{ass:regime}
Fix an active rule regime $g$. Throughout $g$ the maps are constant,
$(W_t,B_t,K_t)\equiv(W,B,K)$, the index rule is exactly affine
$i_t=Wm_t+Bz_t+\eta_t$, and $K$ is nonsingular on marked components. A change in
median selection, clamp, fallback, or stale-feed filter starts a new regime $g'$
with its own $(W',B',K')$ and enters as a separate observation rather than as
nonlinearity within $g$.
\end{assumption}

The model is a networked operator. The realized mark is the equilibrium
$(I-W)^{-1}Bz$ of the whole reference graph rather than a reading of any single
edge. A directed $W$-cycle---for example Gate$\leftrightarrow$Hyperliquid---can
sustain an endogenous fixed point with no external input, the configuration the
identification results below must contend with.

Before asking what an observer can recover, we must check that the closed-window
mark is a well-defined object at all. Equation~\eqref{eq:substituted} is a
feedback system---each venue's mark feeds the index that others reference---so
the reported mark is the solution of a fixed-point equation rather than a free
variable. It establishes existence, uniqueness, and stability, reads the operator
as a convergent path sum, and shows funding sets convergence speed, not the target.

\begin{lemma}[Existence, uniqueness, stability]\label{prop:existence}
Under Assumption~\ref{ass:regime}, the deterministic fixed point of
equation~\eqref{eq:substituted} solves $(I-W)m^*=Bz+\eta$; it is unique iff $1$ is not an
eigenvalue of $W$, giving $m^*=(I-W)^{-1}(Bz+\eta)$, and is locally exponentially
stable iff every eigenvalue of $K(I-W)$ has positive real part.
\end{lemma}
The argument is immediate: setting the drift of equation~\eqref{eq:substituted} to zero and
cancelling the nonsingular $K$ leaves $(I-W)m^*=Bz+\eta$, uniquely solvable exactly
when $I-W$ is invertible, and the error system $\dot x_t=-K(I-W)x_t$ gives the
stability condition. The content is not the inversion but its reading: the mark is
the equilibrium of a network, so one observed level can be produced by many
different $(W,B)$ splits---the identification problem of
Section~\ref{sec:identification}.

The first consequence is a mechanism-design warning specific to perpetuals.
Because $K$ enters equation~\eqref{eq:substituted} only as a left factor of the drift, it
cancels when the drift is set to zero: the fixed point $m^*=(I-W)^{-1}(Bz+\eta)$
does not depend on $K$ at all, while the error dynamics $\dot x_t=-K(I-W)x_t$ only
run faster as $K$ grows. Funding, liquidation, and mark-control intensity
therefore set the \emph{speed} at which the mark reaches the oracle target, never
\emph{which} target it reaches---the closed-window target is a property of the
reference topology, not of how hard the venue enforces its mark. This is why a
venue cannot ``fix'' a self-referential oracle by funding more aggressively:
faster funding reaches the \emph{same} endogenous target sooner, and if that
target is mis-anchored, stronger control merely makes the system efficiently
wrong---the opposite of the reassurance tight mark control is usually taken to
provide.

When $\rho(W)<1$ the operator expands as a convergent Neumann series
$m^*=\sum_{n\ge0}W^n(Bz+\eta)$, so $W^nBz$ is the external anchoring carried into
marks along all length-$n$ reference paths and $E_{\mathrm{eff}}=(I-W)^{-1}B$ is the
effective external exposure that aggregates them. A venue with small \emph{direct}
vendor weight $B$ can therefore still track a proxy closely when its peers import
the same proxy, so direct anchoring and long-path transmission produce the same
mark level while differing in how they amplify shocks---which is how cross-venue
agreement forms without any venue reading an external price directly.

\section{Identification and Reopen Validation}
\label{sec:identification}

In this section $R$ denotes a population reduced form inside one active affine
regime $g$: either the deterministic relation or the conditional mean
$\E[m\mid z,g]=a_g+R_gz$ when $\E[\eta\mid z,g]=0$, written after augmenting
$z$ with a constant anchor or centering, so intercepts are not suppressed.
The regime label is a rule branch fixed before estimating $R_g$; if fallback,
median, clamp, or stale-feed labels are unobserved, row conclusions are
conditional diagnostics, and regime switches are projection residuals, not
hidden evidence for a topology. Regime indices are suppressed for readability.

\subsection{Price data cannot identify anchoring}

The first question is whether an analyst with unlimited closed-window price and
proxy data can decide how a mark is formed. Fix one active affine regime and stack
the anchors a venue may reference---matched external proxies, same-underlying
derivative marks, stale levels---in the vector $z\in\R^Q$. What the data reveal is
the \emph{reduced form} $R\in\R^{V\times Q}$, the matrix sending anchors to the $V$
venue marks through $m=Rz$. Behind it sits the structural oracle pair $(W,B)$: the
peer-topology matrix $W$, whose entry $w_{ij}$ is the weight venue $i$ places on
venue $j$'s mark, and the external-loading matrix $B$, so the marks are the fixed
point $m=(I-W)^{-1}Bz$. Pure external anchoring is $W=0$ (then $R=B$); a
self-referential oracle has $W\neq0$ and routes anchors through the network. The
question is whether $R$ tells these two apart, and the answer is no.

Intuitively, forming a closed-window mark is like setting a wall of clocks: each
clock can be set by a real external time signal (external anchoring, $W=0$) or by
copying its neighbours (peer reference, $W\neq0$). When the clocks agree, their
faces alone---like the marks alone---cannot reveal which mechanism set them, because
a copying network and an external tracker can display identical readings. The
reduced form $R$ is exactly those clock faces, and the results below make this
impossibility, and its disclosure-based cure, precise.

\begin{theorem}[Disclosure-free non-identification]\label{thm:obs-eq}
The reduced form $R$ does not identify anchoring: the externally anchored model
($W=0$) and infinitely many oracle topologies ($W\neq0$) generate the same observed
mapping $m=Rz$.
\end{theorem}

\begin{proof}
Substituting $B=(I-W)R$ into the fixed point gives
\[
  (I-W)^{-1}Bz=(I-W)^{-1}(I-W)Rz=Rz=m,
\]
so $(W,(I-W)R)$ generates the observed mapping for every admissible $W$. The set
$\{W:\det(I-W)\neq0\}$ is open and dense in $\R^{V\times V}$, so the fibre
$\{(W,B):(I-W)^{-1}B=R\}$ is infinite. A mechanism class that imposes sign,
support, row-sum, normalization, or documented-input restrictions on $(W,B)$
intersects this fibre with those restrictions; the reduced form alone performs no
such intersection. \qed
\end{proof}

Theorem~\ref{thm:obs-eq} is static; a dynamic analyst could still hope that
adjustment speeds, lead-lag patterns, or innovation covariances break the tie. The
next theorem removes that hope by reparameterizing the adjustment matrix so the
entire conditional path law is preserved.

\begin{theorem}[Path-law impossibility]\label{thm:path-law}
The full conditional law of $\{m_t\}$ given $\{z_t\}$ still does not identify the
topology: for every $W_2$ with $I-W_2$ nonsingular there is a reparameterization
$(W_2,B_2,K_2,\eta^{(2)})$ of the mark dynamics that induces the same conditional
law, so every measurable functional of the joint law of $(\{m_t\},\{z_t\})$ is
constant across the class.
\end{theorem}

\begin{proof}
Take any $(W_1,B_1,K_1,\Sigma)$ in the class with $I-W_1$ and $K_1$ nonsingular and
the residual law unrestricted, and for a nonsingular $I-W_2$ set
\[
  K_2=K_1(I-W_1)(I-W_2)^{-1},\quad
  B_2=(I-W_2)(I-W_1)^{-1}B_1,\quad
  \eta^{(2)}_t=(I-W_2)(I-W_1)^{-1}\eta_t .
\]
Under equation~\eqref{eq:substituted} the resulting drift has slope and intercept
\[
  K_2(I-W_2)=K_1(I-W_1),\qquad K_2B_2=K_1B_1,\qquad K_2\eta^{(2)}_t=K_1\eta_t ,
\]
so the affine drift and the diffusion $\Sigma$ agree with those of $(W_1,B_1,K_1)$
at every state. The two parameterizations therefore solve the same linear SDE, and
uniqueness in law gives equal conditional laws of $\{m_t\}$ given $\{z_t\}$;
$\mathcal L(\eta^{(2)})$ is admissible because the class leaves the residual law
free. \qed
\end{proof}

Taking $W_2=0$ shows every admissible network is path-law equivalent to a purely
external model, so a test of external anchoring against endogenous consensus built
from lead-lag, Granger, common-factor, or information-share statistics has power
equal to its size: even continuous, error-free path data cannot reject the wrong
mechanism. The class shrinks only when the model restricts $K$ (for example
diagonal venue-specific adjustment) or the support and signs of $(W,B)$, and each
such restriction is a disclosure (Theorem~\ref{thm:row-id}), not a quantity
recoverable from the path law. This is the precise sense in which the closed-window
anchoring question is unanswerable from prices and answerable only from what a
venue discloses about its rule.

\subsection{Disclosure restores identification}

If price data cannot separate the topologies, the question becomes what minimal
piece of rule information does. The cheapest disclosure acts on the adjustment
channel: a venue states that its mark control is \emph{diagonal}---each venue pulls
only its own mark toward its own index, the rule almost every venue actually runs.
We keep the natural regularity for this class: the adjustment $K$ is diagonal with
positive entries, $W\ge0$, the fixed point is locally stable ($K(I-W)$ a nonsingular
M-matrix), and residuals satisfy $\E[\eta_t\mid\mathcal F^{m,z}_t]=0$. Under that one
disclosure the impossibility of Theorem~\ref{thm:path-law} dissolves: the path law
identifies the normalized topology and its support.

\begin{theorem}[Identification under disclosed diagonal adjustment]\label{thm:diag-k}
With $K$ disclosed diagonal and the regularity above, the path law identifies the
drift coefficient $A=K(I-W)$ and the product $KB$. Hence for every venue $i$ the
normalized peer row and its support are identified,
\[
  \frac{w_{ij}}{1-w_{ii}}=-\frac{A_{ij}}{A_{ii}}\quad(j\neq i),
\]
and disclosing $w_{ii}$ (for example $w_{ii}=0$) or $K_i$ identifies the full row
$w_{i\cdot}$ and $b_{i\cdot}=K_i^{-1}(KB)_{i\cdot}$.
\end{theorem}

\begin{proof}
Equality of conditional path laws, with a fixed nondegenerate diffusion observed
continuously, forces equality of the affine drift on the state support. That drift
is
\[
  -K(I-W)\,m_t + KB\,z_t + K\eta_t ,
\]
and since $\E[\eta_t\mid\mathcal F^{m,z}_t]=0$ with $z$ of full support, its slope and
intercept separate, identifying
\[
  A:=K(I-W)\qquad\text{and}\qquad KB .
\]
With $K$ diagonal and $w_{ii}\neq1$ (the M-matrix condition), row $i$ of $A$ is
\[
  A_{ii}=K_i(1-w_{ii}),\qquad A_{ij}=-K_i w_{ij}\ (j\neq i),
\]
so the ratio cancels the unknown speed $K_i$,
\[
  \frac{w_{ij}}{1-w_{ii}}=-\frac{A_{ij}}{A_{ii}},\qquad
  w_{ij}\neq0\iff A_{ij}\neq0 .
\]
Disclosing $w_{ii}=0$ gives $K_i=A_{ii}$ and $w_{ij}=-A_{ij}/A_{ii}$; disclosing
$K_i$ gives $w_{i\cdot}$ from row $i$ of $I-K^{-1}A$; either way $b_{i\cdot}$ follows
from $KB$. \qed
\end{proof}

The economic content is sharp: a one-line disclosure---``each venue adjusts only
its own mark''---converts an estimand with power equal to size into a fully
identified normalized topology, and pairwise drift lead-lag coefficients regain the
content Theorem~\ref{thm:path-law} provably denied them in the unrestricted class.
Such a disclosure is cheap and externally checkable, yet it is exactly what
separates an auditable oracle from an opaque one. The next result turns from the
adjustment channel to the anchor channel: what a single venue must disclose about
its components for its row to be identified, and when the same disclosure
\emph{falsifies} the rule.

\subsection{The row-level disclosure frontier}

The two results above are global, but a venue discloses its rule one row at a time.
For venue $v$, write $P_v$ for the disclosed set of permitted peer inputs, $F_v$ for
the external anchors it forbids ($b_{v,q}=0$ for $q\in F_v$), and let any disclosed
weight or cap limits enter as convex constraints $C_v^Ww+C_v^Bb_v(w)\le d_v$ on the
peer row $w=w_{v,P_v}$, where $b_v(w)=r_v-w^\top R_{P_v,\cdot}$ is the external
loading the row implies. Collecting these gives the row's admissibility set, whose
dimension counts exactly how much ambiguity the disclosure leaves.

\begin{theorem}[Local disclosure frontier for oracle rows]\label{thm:row-id}
The row-admissibility set is
\[
  \mathcal A_v(R)=\bigl\{w\in\R^{|P_v|}:\ R_{P_v,F_v}^{\top}w=r_{v,F_v},\
  C_v^Ww+C_v^Bb_v(w)\le d_v\bigr\},
\]
and, absent active bounds,
\[
  \dim\mathcal A_v(R)=|P_v|-\rank\!\bigl(R_{P_v,F_v}^{\top}\bigr).
\]
\end{theorem}

\begin{proof}
Multiplying $R=(I-W)^{-1}B$ by $I-W$ gives $B=(I-W)R$, so the row is
\[
  b_v=r_v-w_{v,P_v}^{\top}R_{P_v,\cdot}.
\]
The forbidden-anchor restriction $b_{v,q}=0$ ($q\in F_v$) is then the linear system
$R_{P_v,F_v}^{\top}w=r_{v,F_v}$, which intersected with the disclosed convex
constraints is $\mathcal A_v(R)$; rank-nullity gives its dimension. If the kernel of
$R_{P_v,F_v}^{\top}$ is nonzero and unconstrained, any $\delta$ in it gives a second
feasible row: change only row $v$ of $W$, set $B_\varepsilon=(I-W_\varepsilon)R$, and
the reduced form is reproduced, so the row is not identified. \qed
\end{proof}

The set therefore lands in exactly three cases: empty---the reduced form and the
disclosed rule cannot both hold in the local affine model; a singleton---the row
$w_{v,P_v}$ and $b_v$ are locally identified; or positive-dimensional---distinct
oracle rules share the same forbidden anchors and reduced form. The rank count is
the operational lever: a venue permitting $p$ peer inputs must disclose $p$
independent exclusions, weights, or normalizations for its row to be identified,
and a global $W$ is recovered only when every row is a singleton and the completed
$W$ satisfies $1\notin\operatorname{spec}(W)$. More price data never changes this
count; only rule information does.

For an estimated reduced form $\widehat R$ the population equalities are replaced by
the row projection distance
\[
  \Delta_v(\widehat R)=\min_{C_v^Ww+C_v^Bb_v(w)\le d_v}
  \bigl\|\widehat R_{P_v,F_v}^{\top}w-\widehat r_{v,F_v}\bigr\|_2,
\]
with $\mathcal A_v^\tau$ replacing the equalities by residual at most $\tau$. Only
such a tolerance turns the diagnostic into finite-sample compatibility: an
exact-empty $\mathcal A_v(\widehat R)$ is a projection residual, not a population
rejection, unless it exceeds the tolerance the next corollary makes explicit.

\begin{corollary}[Finite-sample tolerance]\label{cor:rank-tolerance}
Let $M_v=R_{P_v,F_v}^{\top}$ and suppose a population-compatible row $w_0$ with
$\|w_0\|_1\le L$ stays feasible for the sample projection problem and the entrywise
coefficient error on rows $P_v\cup\{v\}$ and columns $F_v$ is at most
$\varepsilon$. Then
\[
  \Delta_v(\widehat R)\le\sqrt{|F_v|}\,(L+1)\varepsilon .
\]
\end{corollary}

\begin{proof}
Evaluate the projection objective at the feasible $w_0$. For each $q\in F_v$,
\[
  \bigl|\widehat R_{P_v,q}^{\top}w_0-\widehat r_{v,q}\bigr|
  \le\|w_0\|_1\,\varepsilon+\varepsilon\le(L+1)\varepsilon,
\]
by the entrywise bound and $\|w_0\|_1\le L$; taking the $\ell_2$ norm over the
$|F_v|$ forbidden coordinates gives the bound. \qed
\end{proof}

The corollary calibrates the falsification: an empty exact set is evidence against
a disclosed row only when $\Delta_v(\widehat R)$ exceeds
$\sqrt{|F_v|}\,(L+1)\varepsilon$, so a venue that genuinely follows its rule cannot
be rejected by sampling noise alone, while one whose row is incompatible with $R$
beyond that radius is. Together with Theorem~\ref{thm:row-id} this is the
disclosure frontier an auditor walks: count the independent restrictions to decide
whether the row can be identified at all, then compare $\Delta_v$ to the tolerance
to decide whether the disclosed rule survives.

The non-identification results say price data cannot separate anchoring from
consensus, but the cash reopen can: the moment the primary market reopens it prints
a tradable level no oracle controls. Write the reopen price as $S_o=h^\top z+\nu_o$,
splitting it into the part $h^\top z$ explained by the closed-window external anchors
and the opening-auction residual $\nu_o$, and let $m^-=(I-W)^{-1}(Bz+\eta)$ be the
pre-open mark vector. The reopen gap $G_o=S_o\one-m^-$ then decomposes into an
anchor-mismatch term, a topology-amplified oracle-noise term, and the auction
residual.

\begin{theorem}[Reopen-gap decomposition]\label{thm:reopen}
The pre-open-to-open gap satisfies
\[
  G_o=\left[\one h^\top-(I-W)^{-1}B\right]z-(I-W)^{-1}\eta+\one\nu_o .
\]
\end{theorem}

\begin{proof}
Substitute $S_o=h^\top z+\nu_o$ and $m^-=(I-W)^{-1}(Bz+\eta)$ into
$G_o=S_o\one-m^-$:
\[
  G_o=\one(h^\top z+\nu_o)-(I-W)^{-1}Bz-(I-W)^{-1}\eta,
\]
and collecting the terms in $z$, $\eta$, and $\nu_o$ gives the stated
expression. \qed
\end{proof}

The bracketed term $\one h^\top-(I-W)^{-1}B$ is the validation object: a
well-anchored oracle makes it small because closed-window anchors reproduce the
reopen level, so a large residual gap loading on own and peer drift rather than
external anchors is evidence of endogenous transmission.

\section{Numerical illustration}
\label{sec:simulation}

The identification results are easiest to read on a controlled example where the
true topology is known. We build two oracle networks on $V=3$ venues, each fed by
three anchors---a matched external proxy, a derivative anchor, and a stale level.
The first network is purely external ($W=0$): each venue copies its own index. The
second routes the same external information through a heavy directed peer cycle of
spectral radius $\rho(W)=0.77$, choosing $B=(I-W)R$ so that both networks share the
\emph{same} reduced form $R$ by construction. The example then asks two questions:
can any price statistic tell the two topologies apart, and what does one line of
disclosure buy.

Consider the first question. Under both topologies a representative venue's mark
path coincides to machine precision (maximum gap $1.1\times10^{-13}$) even though
the weight matrices are structurally opposite---a zero matrix against a dense
directed cycle. Every price-based statistic is consequently identical across the two
wirings: the reduced forms differ by $0$, the proxy $R^2$ is $0.9997$ in both,
cross-venue agreement is $0.970$ in both, and a lead-lag information-share proxy is
$0.007$ in both. An analyst with unlimited price and proxy data sees a single object
and cannot recover which mechanism produced it, exactly as
Theorems~\ref{thm:obs-eq}--\ref{thm:path-law} predict.

Figure~\ref{fig:simfrontier} visualizes the disclosure frontier for one venue. The
two axes are that venue's peer weights $(w_{12},w_{13})$---how much its mark loads
on each of the other two venues---so every point in the plane is one candidate row.
The shaded region collects the rows that reproduce the observed reduced form: with
price data alone it is the whole plane, and the row is completely unidentified
(dimension $2$). Disclosing that a derivative anchor is forbidden cuts the plane to
the solid blue line (dimension $1$), and adding the diagonal-adjustment rule of
Theorem~\ref{thm:diag-k} pins that line to the single black marker at the true row
$w^\star=(0.3,0.4)$ (dimension $0$); the collapse $2\to1\to0$ is exactly the count
$|P_v|-\rank(M_v)$ of Corollary~\ref{cor:rank-tolerance}. The red dashed line is a
\emph{false} disclosure---a venue forbidding a derivative anchor it actually
loads---and it never meets the shaded set, so its projection distance is
$\Delta_v=0.11$ instead of the $\Delta_v=0$ of the true rule. The same disclosure
that identifies an honest row therefore falsifies a dishonest one.

\begin{figure}[t]
\centering
\includegraphics[width=0.85\textwidth]{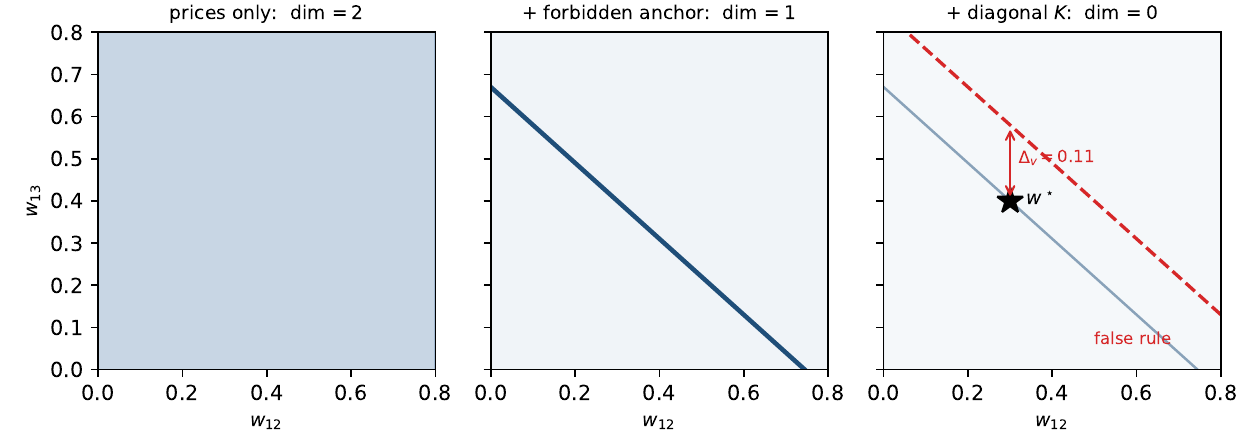}
\caption{Disclosure frontier for one venue. In the peer-weight plane
$(w_{12},w_{13})$ the admissible rows shrink from the shaded plane (prices) to a line
(forbidden anchor) to the true point (diagonal $K$); a false rule (dashed) leaves
$\Delta_v>0$.}
\label{fig:simfrontier}
\end{figure}

Read together, the example is the paper in miniature: across opposite topologies
every price statistic matches---reduced-form gap $0$, proxy $R^2=0.9997$,
cross-venue agreement $0.970$, lead-lag share $0.007$---yet disclosure collapses the
admissible dimension $2\to1\to0$ while a false rule stays infeasible
($\Delta_v=0.11$). Prices cannot identify the oracle; only disclosure, one
restriction at a time, can. The run is reproduced under a fixed seed by
\texttt{figures/make\_simulation.py}.

\section{Measurement Implications}
\label{sec:measurement}

The non-identification theorem dictates the design: lead not with an unconstrained
proxy regression but with restrictions that break the equivalence class. Where the
simulation of Section~\ref{sec:simulation} varies a known topology, real data can
only probe observable signatures, since prices alone do not identify the mechanism
(Theorem~\ref{thm:obs-eq}); each check below offers indirect support for a
proposition rather than a recovery of $W$, and all are reproduced from frozen
artifacts.
Three reproducible companion simulations probe robustness beyond the indirect
real-data signatures---a variance-ratio endogeneity detector, a permutation-based
reopen event study, and an oracle-incident echo illustration---all seeded and
produced by \texttt{figures/make\_experiments.py}.

The theorems yield four observable implications that organize the empirics: row
falsification through $\Delta_v$ (Theorem~\ref{thm:row-id},
Corollary~\ref{cor:rank-tolerance}); limited discrimination, since a
disclosed-weight row and a pure-external $w=0$ baseline fit comparably under the
path law (Theorem~\ref{thm:path-law}); a bounded live-external share, as deep-closed
marks track index futures but not a crypto placebo; and reopen anchoring
(Theorem~\ref{thm:reopen}). Each subsection below reports one, all from frozen public
data.

\subsection{Data}
\begin{definition}[Anchoring type and provenance coding]\label{def:anchoring}
Row $v$ is \emph{externally anchored} if $W_{v\cdot}=0$ and $B_{v\cdot}\neq0$ and
shows \emph{derivative-network anchoring} if $W_{v\cdot}\neq0$; an input enters $B$
only if produced independently of the venue's same-underlying derivative marks and
timestamped before the mark update, and $W$ otherwise, with provenance-ambiguous
inputs (tokenized spot, validator oracle, external perpetual) carried under both
codings.
\end{definition}

We collect frozen public five-minute candles: venue mark and
index series from Binance, Bitget, Gate, and OKX; cash and extended-hours equity
bars from Yahoo; and ES, NQ, and BTC futures for the comovement placebo. The
sample spans April--June 2026 for four underlyings (AAPL, NVDA, TSLA, AMZN), with
each endpoint snapshot frozen so an offline replay reproduces every metric under a
fixed CSV SHA-256 hash. Disclosed OKX component rows are captured on two dated
snapshots, and provenance-sensitive inputs are carried under the dual coding of
Definition~\ref{def:anchoring}.

\subsection{Pre-open row falsification}

A row-level Theorem~\ref{thm:row-id} test estimates $\widehat R$ on a fixed
timestamp grid, builds the disclosed $P_v$ and $F_v$, and reports
$\Delta_v(\widehat R)$ with a fixed-design residual bootstrap for the
Corollary~\ref{cor:rank-tolerance} tolerance and a moving-block bootstrap that
re-estimates $\widehat R$ each replicate.

In the pre-open regime, where extended-hours anchors update live, the OKX row is
tested for four underlyings on $186$--$195$ five-minute timestamps over three NYSE
days under two frozen rule captures. The disclosed rule is falsified nowhere---cap
projection $\Delta_v\in[0.0024,0.0886]$, fixed-weight residuals $0.0036$--$0.1261$,
and a circularity-robust no-own-mark variant $0.0040$--$0.0505$---but the
pure-external $w=0$ baseline residuals ($0.0151$--$0.1275$) are of the same order
with mixed ordering across underlyings (Figure~\ref{fig:preopen}), so on these
rank-deficient, proxy-fed data the check is a falsification the disclosed support,
cap, and weights pass, not an identification (Theorem~\ref{thm:path-law}).

\begin{figure}[t]
\centering
\includegraphics[width=0.6\textwidth]{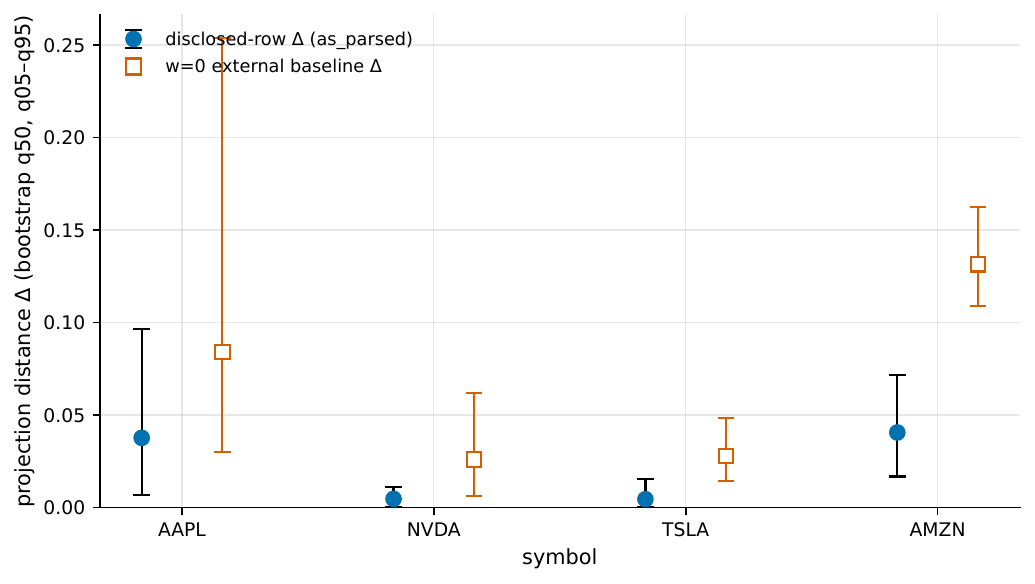}
\caption{Pre-open row projection distance $\Delta_v$ (disclosed row) with
moving-block bootstrap whiskers, against the $w=0$ baseline.}
\label{fig:preopen}
\end{figure}

\subsection{Deep-closed windows at scale}

The deepest closure---overnight, weekends, and one full-day holiday---anchors an
eight-week archival panel (11 April--9 June 2026; $9{,}809$ deep-closed timestamps
per underlying across six venue series). Marks are not frozen (median $20$--$123$ bps
from the last cash close), yet they track external index futures: deep-closed mark
returns on ES/NQ give $R^2$ up to $0.42$ with a BTC placebo below $0.07$. The
endogenous channel is nonetheless visible where no cash or futures information
flows---cross-venue dispersion is \emph{wider} than the regular session and a
Saturday variance floor follows the disclosed rules---the signature a seeded
variance-ratio simulation reproduces ($\mathrm{Var}(m)/\mathrm{Var}(z)>1$ under a
network). The deep-closed drift matches the eventual cash-open sign in $156$ of
$192$ series-reopens, and the frontier separates venues from rule text alone: Gate
discloses constituents but no weights, leaving its row admissibility set at
dimension $1$--$4$, a Theorem~\ref{thm:row-id} non-identification verdict exactly as
Theorems~\ref{thm:obs-eq}--\ref{thm:path-law} require.

\subsection{Archival reopen event study}

Theorem~\ref{thm:reopen} is tested with archival reopen event studies for four
underlyings over five NYSE opens ($104$ venue-series rows). Venue reopen moves track
the matched cash move closely (pooled median transition error $4.49$ bps) and far
better than a wrong-underlying control ($28.68$ bps) or a same-underlying time-shift
placebo ($16.9$--$142.0$ bps); a seeded permutation study on simulated reopens
shows the matched-anchor advantage rejects a no-anchoring null at scale ($p\approx2\times10^{-4}$ in simulation). These support
matched-underlying reopen anchoring without identifying $W$; the topology
falsification stays descriptive (linked pairs more dispersed, $13.39$ vs $6.62$ bps,
Welch $p=0.042$).

\section{Discussion}
The results carry a single mechanism-design message: a 24/7 equity-perpetual venue
cannot demonstrate healthy closed-market price discovery by exhibiting cross-venue
agreement or proxy tracking, because agreement is what an endogenous network
produces and proxy tracking is consistent with both external anchoring and a
network importing the same proxy along long reference paths. What a venue
\emph{can} demonstrate is disclosure---diagonal adjustment, peer support, forbidden
anchors, and weights---together with cash-reopen behavior, the only objects that
move $\mathcal A_v(R)$ off a positive-dimensional class. For a clearinghouse or
regulator this is a disclosure checklist rather than a price test: the row frontier
says how many independent exclusions a venue must publish for its row to be
identified, the finite-sample tolerance says when an apparent incompatibility is
sampling noise, and a venue disclosing enough to make $\mathcal A_v(R)$ a singleton
converts an opaque mark into an auditable one.

\section{Limitations}
The empirical claims are bounded. The samples are small---five reopen days, eight
to nine multi-day-closure reopens, and a thirteen-timestamp overnight pilot---so
inferential statistics are reported as descriptive ranking indices, and several
comovement and reopen checks are contemporaneous, confirming that marks move with
and snap to anchors the oracle already includes rather than discovering them. The
row diagnostics proxy Hyperliquid and own-perpetual components with public
perpetual candles and cannot observe stale-feed or active-branch labels, so
falsification survival is compatibility under a closure assumption, not a recovered
topology; and the venue rules are date-sensitive help-center disclosures that pin
permitted inputs but not fallback weights, which is exactly why identification
rests on disclosure and cash-reopen validation rather than price data alone.

\section{Conclusion}

Closed-market equity perpetual marks should be modeled as oracle fixed points:
$(I-W)^{-1}B$ separates external anchoring from recursive derivative-network
transmission, and within the unrestricted adjustment class no functional of
observed mark and proxy paths can distinguish price discovery from endogenous
consensus, while disclosed diagonal adjustment restores the normalized topology and
disclosed support and forbidden anchors give a row-level frontier. Empirically, the
disclosed OKX row survives pre-open falsification for four underlyings with the
pure-external baseline confirming the check's limited power, sampled reopen
behavior is anchored to the matched cash underlying, and an eight-week deep-closed
panel shows marks tracking live index futures with a Saturday variance floor
aligned to disclosed rules. The mechanism-design implication is that a 24/7
equity-perpetual venue cannot demonstrate healthy closed-market price discovery by
exhibiting cross-venue agreement or proxy tracking; it must disclose enough of the
oracle topology, or pass cash-reopen validation, to break the fixed-point
equivalence class.

\paragraph{Declarations.}
No human-subject data are used. The pilot uses aggregate computed outputs and
anonymized source snapshots. The author(s) report no conflict of interest. AI
assistance was used for drafting and formatting; all claims, proofs, and
submission decisions remain the author(s)' responsibility.

\bibliographystyle{splncs04}
\bibliography{references}

\end{document}